\documentclass[runningheads]{llncs}

\usepackage{graphicx}
\usepackage{xcolor, colortbl}
\usepackage{amsmath,amssymb,amsfonts}
\usepackage{textcomp}

\usepackage{graphicx}
\usepackage{xspace}
\usepackage{todonotes}
\usepackage{subcaption}

\usepackage{multicol}
\usepackage{multirow}

\usepackage{listings}

\usepackage{algorithmicx}
\usepackage{algcompatible}
\usepackage{algpseudocode, algorithm}

\usepackage{authblk}
\usepackage{enumitem}

\usepackage[font=small,skip=0pt]{caption}

\definecolor{olivegreen}{RGB}{0,153,0}

\newcommand{\fecpra}{{\sc {FusEd-PageRank}}}

\begin{document}

\title{\fecpra{}: Loop-Fusion based Approximate PageRank}

\author{
Shalini Jain*, Rahul Utkoor*, Hemalatha Eedi \linebreak
Sathya Peri, Ramakrishna Upadrasta \linebreak
\email{\{cs15resch11010, cs14btech11037, cs15resch11002\}@iith.ac.in} \linebreak
\email{\{sathya\_p, ramakrishna\}@cse.iith.ac.in}
}

\institute{Department of Computer Science and Engineering, IIT-Hyderabad}

\maketitle

\begin{abstract}

PageRank is a graph centrality metric that gives the importance of each node in a given graph.
The PageRank algorithm provides important insights to understand the behavior of nodes through the connections they form with other nodes. It is an iterative algorithm that ranks the nodes in each iteration until all the node values converge.

The PageRank algorithm is implemented using sparse storage format, which results in irregular memory accesses in the code.
This key feature inhibits optimizations to improve its performance, and makes optimizing the PageRank algorithm a non-trivial problem. 
In this work we improve the performance of PageRank algorithm by reducing its irregular memory accesses.

In this paper, we propose \fecpra{} algorithm, a compiler optimization oriented approximate technique that reduces the number of irregular memory accesses in the PageRank algorithm, improving its locality while making the convergence of the algorithm faster with better accuracy in results.
In particular, we propose an approximate PageRank algorithm using Loop-Fusion.

We believe that ours is the first work that formally applies traditional compiler optimization techniques for irregular memory access in the PageRank algorithm.
We have verified our method by performing experiments on a variety of datasets: LAW graphs, SNAP datasets and synthesized datasets.
On these benchmarks, we have achieved a maximum speedup (vs. \texttt{-O3} optimization) of $2.05\times$, $2.23\times$, $1.74\times$ with sequential version, and $~4.4\times$, $~2.61\times$, $~4.22\times$ with parallel version of \fecpra{} algorithm in comparison with Edge-centric version of PageRank algorithm~\cite{ajayPanyala}.

\end{abstract}
\section{Introduction}
Real-world Networks like social networks, road networks, collaboration networks, etc., represented as graphs, conveyed many features in various scientific and engineering applications. In the current era, graphs are becoming more prominent, with complex and irregular structures. Extracting those significant properties necessitates a standard benchmark to overcome the irrelevant information found in the graphs.

Processing large graphs and finding properties between the nodes is compute-intensive. Centrality metrics help in analyzing large graphs based on the properties of the nodes. Various centrality metrics~\cite{DBLP:conf/icpp/BaderM06a} have been proposed in the literature that focus on various properties of the graphs. Some important metrics are PageRank-centrality~\cite{Page1999ThePC}, Betweenness-centrality~\cite{Freeman1977Set}, Closeness-centrality~\cite{1950ASAJ...22..725B}\cite{RePEc:spr:psycho:v:31:y:1966:i:4:p:581-603}, Degree-centrality~\cite{Bonacich}, EigenVector-centrality~\cite{Page1999ThePC}\cite{Bonacich}, etc.

In this work, we focus on Google's PageRank centrality metric~\cite{Page1999ThePC} to address the primary concerns related to faster convergence, storage efficiency, and data locality when computing on a shared-memory architecture.

PageRank is a standard benchmark that derives the importance of nodes in a graph from neighbor nodes. This fundamental graph algorithm ranks the relative importance of each of the nodes iteratively based on the formula mentioned in  Eqn.~\ref{eq}.
\begin{equation}
    pr_u = \frac{1\!-\!d}{n} + d * \sum_{v \in inlink(u)}\frac{pr_v}{d_o(v)}
     \label{eq}
\end{equation}
The node $u$'s PageRank $pr_u$ in $(i+1)^{th}$ iteration is computed from the PageRank values of its inlink vertices $inlink(u)$ from $i^{th}$ iteration divided by its out degree $d_o(v)$. A dampening factor $d$ assigns a uniform probability in jumping to any page, thus ensuring strong connectedness.

In the PageRank algorithm, each node of a graph is initialized with some random value. As the algorithm proceeds iteratively, the difference between the current PageRank value and the previous PageRank value decreases. The decrease in the error value of each node follows the convergence property. When the difference between the PageRank values from two consecutive iterations becomes less than the threshold value, the change in the PageRank value becomes negligible.

The storage of graphs plays an essential role in how we process the data in the PageRank algorithm. Two available storage formats for graphs are adjacency matrix and adjacency list~\cite{Cormen2001introduction}. An adjacency matrix storage format is not recommended in shared-memory architectures with large graphs because of its size limitations and data sparsity.

Consequently, we focus only on sparse-matrix based adjacency list storage format. The storage formats for sparse-matrices are broadly categorized as two types. 1) Vertex-centric, and 2) Edge-centric. 

In the Vertex-centric PageRank computation~\cite{DBLP:conf/sigmod/MalewiczABDHLC10}\cite{DBLP:conf/uai/LowGKBGH10}\cite{DBLP:conf/cidr/WangXDG13}\cite{DBLP:conf/sosp/NguyenLP13} the max size of the data structure used is $\mathcal{O}(V)$, where $V$ is the total number of vertices.
Similarly, in Edge-centric versions~\cite{ajayPanyala}\cite{8091048}, the max size of the data structure used is $\mathcal{O}(V+E)$, where $E$ is the total number of edges. The Edge-centric versions show the improvement in performance by overcoming the poor locality of vertex-centric PageRank computational model. 
We consider the edge-centric algorithm proposed by Panyala et al.~\cite{ajayPanyala} as our baseline and apply loop-fusion compiler optimization, along with reordering of the statements.

In a memory-efficient storage format, accessing the neighboring nodes involve what are termed as \textit{irregular memory accesses}, ones like \texttt{a[b[i]]}. Traditional compilers disable loop-fusion on the loops having irregular memory accesses.
In this work, we propose a loop-fusion technique for the loops with irregular memory accesses.

Loop-fusion~\cite{DBLP:books/crc/Compiler07} is a compiler optimization that works by merging two or more loops into one. Loops are suitable candidates for fusion if they share the same iteration space with no data dependences. It avoids the overhead of loop control structure and may increase the instruction-level parallelism and improve locality, ultimately improving overall performance. Standard compilers like LLVM~\cite{llvm}, GCC~\cite{10.5555/3133373}, ICC~\cite{IccURL}
disable loop-fusion on loops when the memory accesses are irregular to preserve semantic correctness.

We propose \fecpra{} (loop \textbf{Fus}ion based \textbf{Ed}ge-Centric \textbf{PageRank}) an approximate PageRank algorithm with loop-fusion optimization by statement-reordering that improves locality and leads to faster convergence. Our proposed loop-fusion technique follows the principles of the Gauss-Seidel method. Gauss-Seidel method is an improved version of the Jacobi method for solving systems of linear equations iteratively. Jacobi method uses the values from the preceding iteration. In contrast, the Gauss-Seidel method uses the preceding iteration’s values, and the recently updated values from the current iteration.

\vspace{-1.5\baselineskip}
\begin{table}[h!]
\tiny
\begin{center}
\caption{In (a) and (b), we show the Jacobi and Gauss-Seidel examples for a simple linear recurrence equation. In (c) and (d) we show Jacobi and Gauss-Seidel examples for Edge-Centric Pagerank equation}
\begin{tabular}{||c|c||} 
 \hline
 \cellcolor[HTML]{99BADD} (a) Jacobi\cite{10.1145/1052934.1052938}& \cellcolor[HTML]{99BADD} (b) Gauss-Seidel\cite{Gauss-Seidel} \\ 
 \hline\hline
 \begin{lstlisting}[mathescape=true, language=Python]

for $u$ = 1 to $n$:
  for $v \in inlink(u)$:
    $x_u^{k} +\!= d * x_v^{k-1}$
    $\,$ 
    $\,$ 
    $\,$ 
    
\end{lstlisting} & 
\begin{lstlisting}[mathescape=true, language=Python]
for $u$ = 1 to $n$:
  for $v \in inlink(u)$:
    if u $\le$ v
      $x_u^k +\!= d * x_v^{k-1}$
    else 
      $x_u^k +\!= d * x_v^k$
\end{lstlisting} \\
\hline\hline
\hline
\cellcolor[HTML]{99BADD} (c) Edge-centric (Jacobi)~\cite{ajayPanyala} \ & \cellcolor[HTML]{99BADD} (d) \cellcolor[HTML]{99BADD} \fecpra \textbf(Gauss-Seidel) \\
\hline
 \begin{lstlisting}[mathescape=true, language=Python]
for $u$ = 1 to $n$:
  $contrib = d * \frac{x_v^{k-1}}{d_o(v)}$
  for $v \in outlink(u)$:
    $CL(outlink(v)) = contrib$

for $u$ = 1 to $n$:
  for $v \in inlink(u)$:
    $x_u^{k} +\!= CL(v)$
\end{lstlisting} & 
\begin{lstlisting}[mathescape=true, language=Python]
for $u$ = 1 to $n$:
  for $v \in inlink(u)$:
    if $u \le v$
      $x_u^{k} +\!= CL(v)^{k-1}$
    else
      $x_u^{k} +\!= CL(v)^{k}$
  $contrib = d * \frac{x_u^{k-1}}{d_o(u)}$
  for $v \in outlink(u)$:
    $CL(outlink(v)) = contrib$
\end{lstlisting} \\
 \hline
\end{tabular}
\label{table:jacobi_GaussSeidel}
\end{center}
\end{table}
\vspace{-1.5\baselineskip}

In Table~\ref{table:jacobi_GaussSeidel}, we demonstrate the difference between the (a) Jacobi and (b) Gauss-Seidel method using a simple loop that computes a linear recurrence equation; and the loops from our baseline (c) and proposed algorithms (d).

In first row, (a) corresponds to Jacobi updates, where we access $x_v$ values from $(k-1)^{th}$ iteration. In the corresponding Gauss-Seidel version in (b), if $u > v$, we access the $x_v$ values from $k^{th}$ iteration, else from $(k-1)^{th}$ iteration.

Similarly, in the second row, on the left side, the contributions are pre-computed and stored in $CL$ array in the first loop and later used in the second loop for PageRank computation, hence this algorithm follows Jacobi approach.
These two loops executed in phases will restrict users to have a working/correct version of the Gauss-Seidel version.
The right side of the second row represents a fused version (Gauss-Seidel) of the Edge-centric (Jacobi). Here, we compute the PageRank value and immediately write the contribution values to the $CL$ array. Writing to the $CL$ array will allow us to use the updated contribution values in the later computations. If $u > v$, we use the updated $CL$ values from $k^{th}$ iteration, else, we access the $CL$ values from $(k-1)^{th}$ iteration.

Although the applicability of the Gauss-seidel technique on PageRank computation is a proven technique~\cite{Gauss-Seidel},~\cite{DBLP:conf/amcc/SilvestreHS18}. We are the first one to propose fusing the loops when memory accesses are irregular, which enables the latent Gauss-Seidel approximation.

We discuss the complete details of the baseline algorithm in Section~\ref{sec:baseline_algorithm}, and the proposed algorithm \fecpra{} in Section~\ref{sec:proposed_algorithm} of our paper respectively.

The following are the major contributions of our work:
\begin{enumerate}
    \item We apply the Loop-Fusion optimization technique to the edge-centric PageRank algorithm, resulting in a new \fecpra{} algorithm, that has better data re-use and locality characteristics.
    \item Our \fecpra{} algorithm effectively uses the Gauss-Seidel style for updating PageRank values, (instead of the Jacobi style), resulting in a faster convergence.
    \item We prove that the termination properties of our proposed approximate \fecpra{} algorithm  are similar to the termination properties of standard PageRank algorithm.
    \item We present experimental results on a large variety of input graphs and prove the scalability of our proposed technique.
\end{enumerate}

We organize the rest of the paper as follows.
In Sec.~\ref{sec:Related_work}, we discuss some relevant previous works on PageRank algorithm. 
In Sec.~\ref{sec:baseline_algorithm}, we discuss about the edge-centric PageRank algorithm. Sec.~\ref{sec:proposed_algorithm} presents our proposed algorithm, and its convergence properties. 
In Sec.~\ref{sec:proofs}, we discuss the proofs for the correctness and termination conditions of our proposed algorithm. 
In Sec.~\ref{sec:experimentation} we discuss experiments and analysis of the results.
In Sec.~\ref{sec:conclusion}, we conclude our work and discusses future directions.

\section{Related Work}
\label{sec:Related_work}
Google's first algorithm to search for relevant web pages is the PageRank algorithm proposed by Page et al.~\cite{Page1999ThePC}. This standard benchmark is an iterative algorithm that computes the rank of a webpage from the ranks obtained by its neighbour pages linking to it. A page connected with high ranked pages will possess higher PageRank value. Each page is rendered as a \textit{vertex} in a graph and the links between them are represented as the \textit{edges}. 

The Jacobi~\cite{10.1145/1052934.1052938} is a simple stationary iterative method that is best suited to solve linear system of equations when compared with the Gaussian elimination method. 
The PageRank values are computed in the current iteration by accessing the neighbouring nodes PageRank values from the previous iteration, and this process continues until the convergence results.

The Gauss-seidel method is an improved version of the Jacobi iterative method, where the PageRank values are computed from taking the recent values updated in the current iteration. Arasu A. et al.,~\cite{Gauss-Seidel} applied Gauss-seidel iterative method on the PageRank algorithm and proved a significant speed improvement and faster convergence when compared with Jacobi iterative method.  

Real-world graphs have different structural properties that impact the computation of an algorithm. Addressing these challenges, Garg et al.~\cite{STIC-D} proposed STIC-D, i.e., four algorithmic pre-processing techniques to optimize the parallel vertex-centric PageRank computation on real-world graphs. Computing \textbf{S}trongly Connected Components and traversing them in a \textbf{T}opological order leads to faster processing on PageRank computation. Identifying the nodes (\textbf{I}dentical nodes) with similar properties eliminates the redundant PageRank computation is the second optimization technique. 
Bypassing the PageRank computations of the nodes that form \textbf{C}hains in a directed path and eliminating the \textbf{D}ead Nodes that do not make contributions are the other two optimization techniques the authors proposed to accelerate the PageRank computations. These optimization techniques show improvements with real-world datasets; however, they do not take benefits of data locality because of the limitations in the vertex-centric PageRank computational model.

The two primary ways to implement a parallel PageRank algorithm on shared memory architecture are the Blocking mechanism and the Non-Blocking mechanism. The Blocking mechanism deals with blocking techniques like barriers, locks etc. to prevent simultaneous access by multiple threads in writing and reading the PageRank values of a node. The Non-Blocking mechanism allows updating the PageRank values atomically using hardware instructions such as compare-and-swap etc. This Non-Blocking mechanism leads to the approximate behavior of the PageRank algorithm.
Eedi et al.~\cite{Hemlatha.et.al.}, proposed a Non-Blocking implementation of the PageRank algorithm to eliminate barriers and compared their results with the Blocking implementation of Garg et al.~\cite{STIC-D} work. Their method proved the correctness of the algorithm under concurrent data access. The authors also proved the correctness with proper termination conditions.

With the increase in the demand for computational performance, there is a massive demand for approximate computing techniques to process large-scale graphs with complex structures. Panyala et al.~\cite{ajayPanyala} present approximate computing techniques on graph iterative algorithms - PageRank Algorithm and Community Detection. A Loop Perforation technique is applied for computing the PageRank value of nodes by gradually skipping some portion of operations on the neighbor nodes with minor importance from computations. These strategies reduces the number of memory writes and thereby reduces irregular memory access count.

In an iterative graph algorithm, convergence is an essential factor that directly affects the overall computation cost. Silvestre et al.~\cite{DBLP:conf/amcc/SilvestreHS18} proposed the PageRank algorithm based on asynchronous Gauss-Seidel iterations that yield faster convergence when compared with the conventional power iterative method.

The amount of time taken to access the required data from memory plays a major role in processing massive graph datasets with irregular structures and with varied data sizes and cache sizes. Beamer et al.~\cite{DBLP:conf/ipps/BeamerAP17} presented an optimization technique called \textit{propagation blocking} to reduce memory communication thereby improving locality in computing the PageRank algorithm.

\section{Baseline Algorithm}
\label{sec:baseline_algorithm}
This section mainly focuses on two storage formats of the PageRank algorithm: 1) Vertex-Centric, and 2) Edge-Centric. In Table~\ref{table:storage_formats}, we do a comparison of the two storage formats.

\vspace{-1.5\baselineskip}
\begin{table}[h!]
\tiny
\begin{center}
\caption{Comparison of two storage formats}
\begin{tabular}{||l | c | c||} 
 \hline
 \cellcolor[HTML]{99BADD} \textbf{Property} & \cellcolor[HTML]{99BADD} \textbf{Vertex-Centric} & \cellcolor[HTML]{99BADD} \textbf{Edge-Centric} \\ [0.5ex] 
 \hline\hline
 Redundant computations & \cellcolor[HTML]{FC9191}Yes & No \\ 
 \hline
 Poor locality & \cellcolor[HTML]{FC9191}Yes & No \\
 \hline
 Storage space complexity & $\mathcal{O}(V)$ & \cellcolor[HTML]{FC9191}$\mathcal{O}(V+E)$ \\
 \hline
\end{tabular}
\label{table:storage_formats}
\end{center}
\end{table}
\vspace{-1.5\baselineskip}

Let $G(V, E)$ be a directed graph, where input tuple $(u,v) \in E$ represents the edge corresponding to the given graph. 
For any pair of nodes, a link from node $u$ to node $v$ i.e., $u \rightarrow v$ is termed as an $outlink$ from $u$ and $inlink$ to $v$.

In the Vertex-centric PageRank algorithm, for computing the PageRank of each node, we access its respective $inlink$ neighbour's PageRank value and compute the contribution value.

\vspace{-1.2\baselineskip}
\begin{figure}[h]
    \centering
    \includegraphics[scale=0.5]{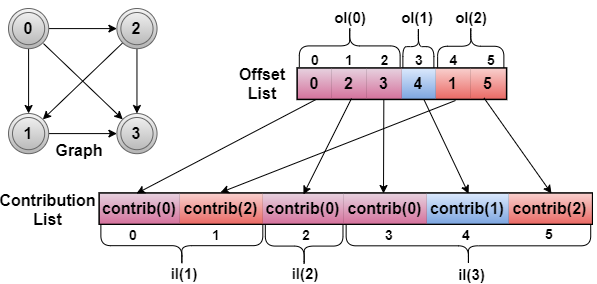}
    \vspace{0.5\baselineskip}
    \caption{Relation between $offsetList$ and $contributionList$}
    \label{fig:offset_contrib_list}
\end{figure}
\vspace{-1.2\baselineskip}

The Edge-centric PageRank algorithm, proposed by Panyala et al.~\cite{ajayPanyala}, divides the PageRank algorithm in two phases.
In the first phase, for each node, the contribution value is computed and stored them in the $contributionList$ array. 
In the second phase, $contributionList$ array is used for computing the PageRank of each node, where we access the contribution values contiguously from $contributionList$ array. The algorithm uses extra memory to store the contribution value of each node to its respective out-linking neighbours.
Pre-computing the contribution values in first phase results in elimination of redundant computations as compared with Vertex-centric model. Also, the contribution values are accessed contiguously while computing the PageRank of each node, which results in improved locality.
Algorithm~\ref{algorithm:PageRankEdge}, is our baseline Edge-centric PageRank algorithm, which uses extra memory for storing contribution values, eliminates redundant computations and improves locality by accessing the contribution values contiguously.

In Fig~\ref{fig:offset_contrib_list} we illustrate the $contributionList$ and $offsetList$ data-structures using the example graph. Here, inlinks of vertex $3$ are $\{0, 1, 2\}$, and we store their contributions in $contributionList$ at index $3, 4, 5$ respectively. Vertex $3$ becomes $outlink$ to nodes $\{0, 1, 2\}$; and we store the indices $3, 4, 5$ at index $2, 3, 5$ of $offsetList$ array. In $contributionList$ we store the contributions of node $inlinks$ contiguously and in $offsetList$ we store the indices of node outlinks which are mapped to $contributionList$ contiguously.

\vspace{-1.2\baselineskip}
\begin{algorithm}[!htb]
\caption{Baseline Algorithm~\cite{ajayPanyala}}
\scriptsize

\begin{algorithmic}[1]
\setlength{\lineskip}{3pt}
\State \textbf{Input:} Graph {G} $\leftarrow$ (V, E)
\Procedure{PageRank}{$G = (V, E)$}      
    \State  $ error \leftarrow 1 $
    \State  $ threshold \leftarrow  10^{-15} $
    \ForAll{ nodes $u_{i}$ $\vert i \in  \lbrace 1, ...,\textit{n}\rbrace$} \Comment{parallel}
        \State $pr(u_{i})\leftarrow \frac{1}{n}$
    \EndFor
    \While{$error > \  threshold$}
    \State $error  \leftarrow 0$
    
    \ForAll{ $u \in allVertices $ } \Comment{parallel}
    \State $contribution \leftarrow \displaystyle\frac{pr(u)}{outDeg(u)}$
    \ForAll{ $v \in outList(u)$ }
        \hspace{2cm} \State $contributionList(offsetList(v)) = contribution$
    \EndFor
    \EndFor
    
   \ForAll{ $u \in allVertices $ } \Comment{parallel}
        \State $prev \leftarrow pr(u) $
    \State $sum \leftarrow 0 $
    \ForAll{ $v \in inList(u)$ }
       \hspace{2cm} \State $sum = sum + contributionList(v)$ 
    \EndFor
    \State $pr(u) = \displaystyle\frac{(1-d)}{n} + (d*sum)$
    
    \hspace{1.5cm} \State$error=max(error,\left| prev-pr(u) \right|)$
    \EndFor
  \EndWhile
\EndProcedure

\end{algorithmic}
\label{algorithm:PageRankEdge}
\end{algorithm} 
\vspace{-1.2\baselineskip}

Algorithm~\ref{algorithm:PageRankEdge} iterates over \textit{while} loop till the termination condition becomes true. When the error value of the algorithm is less than the threshold value, the algorithm terminates. In the \textit{while} loop, the kernel has split into two phases. The first phase computes the contribution values for each node and assigns them to their respective out-going nodes and in the second phase the PageRank of the nodes is computed by fetching the pre-computed contribution values during each iteration. 

In Algorithm \ref{algorithm:PageRankEdge}, $inlist$ and $outlist$ are stored in AOS (Array of structures) format which consists of start and end positions of the contiguous buffers from $contributionList$ and $offsetList$ arrays, respectively.
$contributionList$ is an array, which stores contribution values of inlinks contiguously for each vertex of the graph. The contribution of a vertex is the value computed by dividing PageRank by its outdegree. $offsetList$ is an array, which stores the index values corresponding to $contributionList$ when we iterate over the outlinks of each vertex.

Current compiler heuristics for loop-fusion is generic enough to handle the loops with regular memory accesses. Loops with irregular memory access are considered bad candidates for loop fusion. In Algorithm~\ref{algorithm:PageRankEdge}, both the outer for loops (line $10$ and line $16$) share the same iteration space (iterates over all vertices of the graph), but irregular memory access to the $contributionList$ inside the inner loop (line $13$) prevents both the loops from being fused.

We consider the Edge-centric PageRank algorithm as our baseline algorithm and apply loop-fusion manually on the loops with irregular memory access. The modified algorithm results in Gauss-Seidel approximate version of PageRank algortihm.

In the later sections, we explain the details of our proposed technique, \fecpra{} algorithm and its correctness property.

\section{\fecpra{} Algorithm}
\label{sec:proposed_algorithm}
\subsection{Edge-centric + Loop-Fusion}

We propose Algorithm~\ref{algorithm:ApproxPR}, which is a new version of the PageRank algorithm. In this version, we fuse the loops manually, which are considered ``\textit{bad candidates}'' by the state-of-the-art compiler-heuristics because of their irregular structures. We show that our fusion technique results in faster convergence with improved data locality.

\begin{algorithm}[!htb]
\caption{\fecpra{} Algorithm}
\scriptsize

\begin{algorithmic}[1]
\setlength{\lineskip}{3pt}
\State \textbf{Input:} Graph {G} $\leftarrow$ (V, E)
\Procedure{PageRank}{$G = (V, E)$}      
    \State  $ error \leftarrow 1 $
    \State  $ threshold \leftarrow  10^{-15} $
    \ForAll{ nodes $u_{i}$ $\vert i \in  \lbrace 1, ...,\textit{n}\rbrace$} \Comment{parallel}
        \State $pr(u_{i})\leftarrow \frac{1}{n}$
    \EndFor
    \ForAll{ $u \in allVertices $ } \Comment{parallel}
    \State $contribution \leftarrow \displaystyle\frac{pr(u)}{outDeg(u)}$
    \ForAll{ $v \in outList(u)$ }
       \hspace{2cm} \State $contributionList(offsetList(v)) = contribution$
    \EndFor
    \EndFor
    \While{$error > \  threshold$}
    \State $error  \leftarrow 0$
    
   \ForAll{ $u \in allVertices $ } \Comment{parallel}
    \State $prev \leftarrow pr(u) $
    \State $sum \leftarrow 0 $
    \ForAll{ $v \in inList(u)$ }
       \hspace{2cm} \State $sum = sum + contributionList(v)$ 
    \EndFor
    \State $pr(u) = \displaystyle\frac{(1-d)}{n} + (d*sum)$
    \State $contribution \leftarrow \displaystyle\frac{pr(u)}{outDeg(u)}$
    \ForAll{ $v \in outList(u)$ }
       \hspace{2cm} \State $contributionList(offsetList(v)) = contribution$ 
    \EndFor
    \hspace{1.5cm} \State$error=max(error,\left| prev-pr(u) \right|)$
    \EndFor
  \EndWhile
\EndProcedure

\end{algorithmic}
\label{algorithm:ApproxPR}
\end{algorithm} 

In the first phase of the baseline algorithm, we first compute the contribution values and populate them to the $contributionList$ array and only then we enter the next phase. Whereas, in Algorithm~\ref{algorithm:ApproxPR}, after computing the PageRank of a node $u$, we immediately populate the newly computed contribution value to $u$'s out-neighbours. Updating the contributions right after computing the PageRank will result in erasing the older contribution values. In the later stage of the computation, if node $u$ wants to access the contribution value of a neighbour $v$ whose PageRank value is computed in the current iteration, then $u$ will access the contribution value updated in the current iteration. This property allows a node to access the incoming neighbours PageRank contribution from its previous iteration or the current iteration.

In Algorithm1, “read-from” and “write-to” PageRank array (Line11 and Line22) are from two different loops, so we do not preserve locality. While, in Algorithm2, “write-to” and “read-from” PageRank array (Line22 and Line23) are from two consecutive statements, with same array indices. This implies that the PageRank value is readily available in the register for computing the contribution value, which improves the temporal locality of our fused algorithm.

The difference between the PageRank value from the previous iteration and the current iteration is the node's error. As the algorithm proceeds iteratively, the error of each node decreases. Every node in the algorithm exhibits convergence property. We took advantage of this behaviour to apply the loop-fusion technique. Lemma~\ref{lemma1} gives the proof for the proposed technique.

\subsection{Faster Convergence on PageRank}
When compared with baseline, our proposed algorithm takes fewer iterations for attaining convergence which makes it approximate. In this section, we explain the approximation with an example.

\begin{figure}
    \centering
    \includegraphics[scale=0.6]{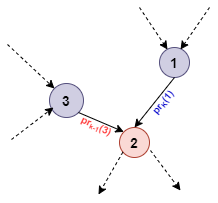}
    \caption{Subgraph showing  computation order for \fecpra{}}
    \label{fig:approximate}
\end{figure}

For example, in Figure~\ref{fig:approximate} the order of execution is, $pr(1)\rightarrow pr(2)\rightarrow pr(3)\rightarrow \cdots$. After computing $pr(1)$ in $k^{th}$ iteration, the $pr(2)$ computation needs to access the PageRank values of nodes \{$pr(1)$, $pr(3)$\}. Here, $pr_{k-1}(3)$ is the value computed from the $(k-1)^{th}$ iteration, while $pr(1)$ is already computed in the $k^{th}$ iteration. Though we use updated $pr_k(1)$ in each iteration, at the end $pr(2)$ converges with the same result as that of baseline.

The convergence of a node depends on the convergence of its neighbouring nodes. In the example shown in Figure~\ref{fig:approximate}, as we already discussed that $pr(2)$ accesses $pr_k(1)$ (updated value) and $pr_{k-1}(3)$ (old value). $pr_k(1)$ implies the node is advanced in the convergence step. Convergence of $pr(2)$ with $pr_k(1)$ is slightly better than with $pr_{k-1}(1)$.

In an iterative algorithm, accessing the previous iteration values and the values updated in the present iteration is termed as approximate technique. Gauss-Seidel method is the first technique that shows the power of this approximation. Our proposed approach is motivated by the Gauss-Seidel method's fundamental principle to use the latest updated values for PageRank computation.

\subsection{Race Conditions in \fecpra{}}
In Algorithm~\ref{algorithm:ApproxPR} there is a read-from contributionList in line 20 and write-to contributionList in line 25, which may lead to data races. 
In such a scenario, a thread can read the previous contribution value (or the one updated in the current iteration).
We use $C++$ vector data structure to store contributions and PageRank values which guarantee thread safety property when multiple threads try to perform operations in parallel. The thread-safety property of the $C++$ vector template is given in this link, https://en.cppreference.com/w/cpp/container.

Although $C++$ vector template guarantees thread safety, we can not achieve Sequential-consistency here as our \fecpra{} may change order of execution when compared with our baseline algorithm execution order. As our proposed algorithm is approximate and may use the updated values from current iteration to compute PageRank of a node, hence we may guarantee algorithm termination and correctness even after violating Sequential-consistency property.

As the data is not properly distributed across the threads, the cache miss rate of our parallel code is high, we are considering improving our current algorithm by reducing the communication cost between the threads and making use of cache effectively

\section{Proofs}
\label{sec:proofs}

In this section, using Lemma~\ref{lemma1}, we explain the characteristics of Loop-Fusion and statement reordering techniques when applied on loops with irregular memory accesses in the PageRank algorithm.
We prove that our proposed \fecpra{} algorithm will lead to a similar termination property as that of the baseline algorithm.

The basic notations for our proof are given below:

\begin{tabular}{rl}
     $n \leftarrow$&Total number of nodes \\
     $d \leftarrow$&Dampening factor ($0.85$) \\
     $\delta \leftarrow$&Threshold value ($10^{-15}$) \\
     $inlink(i) \leftarrow$&List of Inlinking neighbors to node i \\
     $d_o(j) \leftarrow$&Our-degree of node j \\
     $x_i(k\!+\!1) \leftarrow$&PageRank of node $i$ in ${(k\!+\!1)}^{th}$ iteration  \\
\end{tabular}

\begin{lemma}
\label{lemma1}
The algorithm eventually terminates with exact results when applied loop-fusion that yields faster convergence.

\end{lemma}
\begin{proof}
Eqn.~\ref{eq2} represents the Edge-centric PageRank computation of node $i$ in $(k+1)^{th}$ iteration and its termination condition is given by Eqn.~\ref{eq3}.
When the difference between the PageRank values of each node from two consecutive iterations becomes less than the given threshold ($\delta$), then the algorithm terminates.
\begin{equation}
    x_i(k\!+\!1) = \frac{1\!-\!d}{n} + d * \sum_{j \in inlink(i)}\frac{x_j(k)}{d_o(j)}
    \label{eq2}
\end{equation}
\begin{equation}
     \max(|x_i(k\!+\!1) - x_i(k)|) \le  \delta
     \label{eq3}
\end{equation}
In-order to satisfy the termination condition in Eqn.~\ref{eq3}, each node should follow Eqn.~\ref{eq4}. 
When all the nodes satisfy Eqn.~\ref{eq4}, only then the algorithm terminates.
\begin{equation}
 |x_i(k\!+\!1) - x_i(k)| \le \delta
 \label{eq4}
\end{equation}
Now, we analyse the termination property for one node. When we expand Eqn.~\ref{eq4} by substituting the values from Eqn.~\ref{eq2}, the resultant equation becomes,
\begin{equation}
    \bigg |\sum_{j \in inlink(i)}\frac{x_j(k)}{d_o(j)} * d -
    \sum_{j \in inlink(i)}\frac{x_j(k\!-\!1)}{d_o(j)} * d\bigg | \le  \delta
    \label{eq5}
\end{equation}
After expanding the summation from Eqn.~\ref{eq5}, the resultant becomes Eqn.~\ref{eq6}. Re-arranging the terms from Eqn.~\ref{eq6} will give us the following Eqn.~\ref{eq7}.
\begin{equation}
\begin{split}
    &\bigg | d * \bigg (\frac{x_{j_1}(k)}{d_o(j_1)} + \cdots \bigg )
    -\ d * \bigg (\frac{x_{j_1}(k\!-\!1)}{d_o(j_1)} + \cdots \bigg )\bigg | \le  \delta
\end{split}
    \label{eq6}
\end{equation}
\begin{equation}
    d * \frac{|x_{j_1}(k)\!-\!x_{j_1}(k\!-\!1)|}{d_o(j_1)} + d * \frac{|x_{j_2}(k)\!-\!x_{j_2}(k\!-\!1)|}{d_o(j_2)} + \cdots
    \le  \delta
    \label{eq7}
\end{equation}
Here, $d, d_o(j_1), d_o(j_2), d_o(j_3),\ldots, d_o(j_n)$ are constants and the difference between the PageRank values from two consecutive iterations is the variable part. In-order to satisfy the summation in Eqn.~\ref{eq7}, each entity in the summation should also be less than threshold ($\delta$).
\begin{equation}
    d * \frac{|x_{j_1}(k) - {x_{j_1}(k\!-\!1)}|}{d_o(j_1)} \le  \delta
    \label{eq8}
\end{equation}
In Eqn.~\ref{eq8}, $d$ and $d_o(j_1)$ are constants, while  $|x_{j_1}(k)\!-\!{x_{j_1}(k\!-\!1)}|$ varies. For the above inequality to be true, the difference between the PageRank values from two consecutive iterations ($|x_{j_1}(k)\!-\!{x_{j_1}(k\!-\!1)}|$) should decrease iteratively. The decrease in the error value iteratively implies the convergence of each node.

The primary condition that we derive is that each node should exhibit a convergence property. Now, to prove the validity of our proposed \fecpra{} algorithm, we should prove the node convergence property for the modified PageRank Eqn.~\ref{eq9}.
\begin{equation}
  x_i(k\!+\!1) = \frac{1\!-\!d}{n} + d * \bigg ( \sum_{\substack{j \in inlink(i) \\ i<=j}}\frac{x_j(k)}{d_o(j)} + \sum_{\substack{j \in inlink(i) \\ i>j}}\frac{x_j(k\!+\!1)}{d_o(j)}\bigg )
  \label{eq9}
\end{equation}
Eqn.~\ref{eq9} represents our proposed Gauss-Seidel version of PageRank equation. The termination condition of proposed algorithm given by  Eqn.~\ref{eq10} is derived by substituting Eqn.~\ref{eq9} in Eqn.~\ref{eq4}.
\begin{equation}
\begin{split}
    &\bigg |\frac{1\!-\!d}{n} + d * \bigg (\sum_{\substack{j \in inlink(i) \\ i<=j}}\frac{x_j(k)}{d_o(j)} + \sum_{\substack{j \in inlink(i) \\ i>j}}\frac{x_j(k\!+\!1)}{d_o(j)}\bigg ) - \bigg (\frac{1\!-\!d}{n} +\\
    &d * \bigg (\sum_{\substack{j \in inlink(i) \\ i<=j}}\frac{x_j(k\!-\!1)}{d_o(j)} + \sum_{\substack{j \in inlink(i) \\ i>j}}\frac{x_j(k)}{d_o(j)}\bigg )\bigg )\bigg | \le \delta
\end{split}
    \label{eq10}
\end{equation}
After re-arranging the terms in Eqn.~\ref{eq10}, the resultant equation is given by Eqn.~\ref{eq11}.
\begin{equation}
    d * \bigg |\sum_{\substack{j \in inlink(i) \\ i<=j}}\frac{x_j(k)\!-\!x_j(k\!-\!1)}{d_o(j)}\bigg | + d * \bigg |\sum_{\substack{j \in inlink(i) \\ i>j}}\frac{x_j(k\!+\!1)\!-\!x_j(k)}{d_o(j)}\bigg | \le  \delta
    \label{eq11}
\end{equation}
In Eqn.~\ref{eq11}, the first term  
$d * \bigg |\sum_{\substack{j \in inlink(i) \\ i<=j}}\frac{x_j(k)-x_j(k\!-\!1)}{d_o(j)}\bigg |$,
is similar to that of Eqn.~\ref{eq5} for $k^{th}$ iteration. Similarly, second term $d * \bigg |\sum_{\substack{j \in inlink(i) \\ i>j}}\frac{x_j(k\!+\!1)\!-\!x_j(k)}{d_o(j)}\bigg |$ is similar to Eqn.~\ref{eq5} for the $(k+1)^{th}$ iteration.
As Eqn.~\ref{eq5} results in Eqn.~\ref{eq8}, we can conclude that Eqn.~\ref{eq11} also results in Eqn.~\ref{eq8}, which implies our proposed Gauss-Seidel version algorithm with Eqn.~\ref{eq9} will ultimately results in Eqn.~\ref{eq8}.

We prove that the termination property of our proposed algorithm \fecpra{} is similar to that of the standard algorithm.  With our modified equation, we consider nodes that are one step ahead in convergence for error computation, which results in faster convergence of the overall PageRank algorithm. Our experiments also prove that our technique results in the same results and takes fewer iterations to converge than the baseline PageRank algorithm.
Hence, we conclude that our proposed technique is correct. And, with experiments, we prove that our technique results in a faster convergence.

\end{proof}
\section{Experimentation}
\label{sec:experimentation}

\subsection{System model/ Platform}
We ran all our experiments on an Intel-Xeon X5678 machine with 24 CPU cores running at a frequency of 3.07 GHz and 157GB of RAM with 32K, 256K and 12288K of L1, L2 and L3 cache sizes, respectively. We have written all our programs in C++ and compiled them using g++-10 compiler and POSIX multithreaded library by enabling O3 compiler optimization flag in ubuntu 20.04.

\begin{table*}[!htbp]
\tiny
\renewcommand{\arraystretch}{1.3}
\centering
\caption{List of benchmarks with num. of vertices ($|V|$) and num. of Edges ($|E|$) }
\label{tab:benchmark_list}
\begin{tabular}{||c|c|c||c|c|c||c|c|c||}
\hline
\multicolumn{6}{||c||}{ \cellcolor[HTML]{99BADD} \textbf{\textit{ Standard Graphs \cite{LAW}\cite{snapnets}}}} & \multicolumn{3}{c||}{\textbf{\cellcolor[HTML]{99BADD} \textit{Synthetic Graphs \cite{DBLP:conf/sdm/ChakrabartiZF04}}}}\\
        \hline
        \multicolumn{3}{||c||}{\cellcolor[HTML]{8FBC8F} \textbf{\textit{ Web Networks \cite{LAW}\cite{snapnets}}}} & \multicolumn{3}{c||}{\cellcolor[HTML]{8FBC8F} \textbf{\textit{ Social Networks \cite{LAW}\cite{snapnets}}}} & \multicolumn{3}{c||}{-} \\
        \hline
\cellcolor[HTML]{E7DBDB} \textbf{Input} & \cellcolor[HTML]{E7DBDB} \textbf{$|V|$} & \cellcolor[HTML]{E7DBDB} \textbf{$|E|$} & \cellcolor[HTML]{E7DBDB} \textbf{Input} & \cellcolor[HTML]{E7DBDB} \textbf{$|V|$} & \cellcolor[HTML]{E7DBDB} \textbf{$|E|$} & \cellcolor[HTML]{E7DBDB} \textbf{Input} & \cellcolor[HTML]{E7DBDB} \textbf{$|V|$} & \cellcolor[HTML]{E7DBDB} \textbf{$|E|$} \\
\hline
        enwiki-2013 & 4.2M & 101.3M &hollywood-2011 & 2.1M & 228.9M & RMAT-21 & 2M & 41.9M \\
        \hline
        indochina-2004 & 7.4M & 194.1M & twitter-2010-nat & 41.6M & 1468M & RMAT-22 & 4.1M & 83.8M \\
        \hline
        uk-2002-nat & 18.5M & 298.1M & soc-LiveJournal1 & 4.8M & 68.9M & RMAT-23 & 8.3M & 167.7M \\
        \hline
        arabic-2005-nat & 22.7M & 639.9M &soc-Epinions1 & 0.07M & 0.5M & D10 & 0.49M & 1M \\
        \hline
        webbase-2001-nat & 118.1M & 1019.9M & Slashdot0811 & 0.07M & 0.9 & D20 & 0.9M & 2M \\
        \hline
        it-2004-nat & 41.2M & 1150.7M & Slashdot0902 & 0.08M & 0.9M & D30 & 1.4M & 3M \\
        \hline
        sk-2005-nat & 50.6M & 1949.4M &  \multicolumn{3}{c||}{\cellcolor[HTML]{8FBC8F} \textbf{\textit{ Road Networks\cite{LAW}\cite{snapnets}}}} & D40 & 1.8M & 4M \\
        \hline
        web-Stanford & 0.2M & 2.3M & road-italy$\_$osm & 6.6M & 7M & D50 & 2.3M & 5M \\
        \hline
        web-Notre & 0.3M & 1.4M & great-britain$\_$osm & 7.7M & 8.2M & D60 & 2.7M & 6M \\
        \hline
        web-BerkStan & 0.6M & 7.6M & asia$\_$osm & 12M & 12.7M & D70 & 3.2M & 7M \\
        \hline
        web-Google & 0.8M & 5.1M & germany$\_$osm & 11.5M & 12.4M & - & - & -\\
        \hline
\end{tabular}
\end{table*}

\subsection{Dataset}
We verified our proposed algorithm with standard (LAW~\cite{LAW} and SNAP~\cite{snapnets}) and synthetic~\cite{DBLP:conf/sdm/ChakrabartiZF04} benchmarks. We collected graph datasets from reliable sources and categorized into three types, 1) Web Networks, 2) Social Networks, 3) Road Networks. Also, we manually generated the datasets and ran our experiments on synthetic graphs dataset generated from Graph500 R-MAT data generator~\cite{DBLP:conf/sdm/ChakrabartiZF04}. We present the details of the datasets in Table~\ref{tab:benchmark_list}.

\subsection{Results}

In this section, we discuss the experimental results with various metrics for \fecpra{} algorithm wrt Edge-centric algorithm~\cite{ajayPanyala}. We name sequential and parallel versions of the Edge-centric algorithm as Seq-EC and Par-EC. Similarly, we name our sequential and parallel versions of \fecpra{} algorithm as Seq-FusEd and Par-FusEd.

Figure~\ref{fig:combinedPlots} shows our proposed algorithm's effectiveness and approximate behaviour in terms of speedup, iterations, and \textit{L1-Norm} metrics. We show the performance variations by executing the sequential, parallel versions of baseline and \fecpra{} algorithms on the medium-sized benchmarks \textit{enwiki-2013} \cite{LAW}, \textit{indochina-2004} \cite{LAW} from web-networks and \textit{hollywood-2011} \cite{LAW} from social-networks with the varying threshold ranging in $\left [ 10^{-15}, 10^{-06} \right ]$.

Figure~\ref{fig:combinedPlots}a,~\ref{fig:combinedPlots}b and~\ref{fig:combinedPlots}c show the speedup obtained by Par-EC, Seq-FusEd and Par-FusEd algorithms wrt Seq-EC algorithm.
The plots indicate that the Par-FusEd algorithm always performs better than the other versions.

\begin{figure}[h]
    \centering
    \includegraphics[scale=0.28]{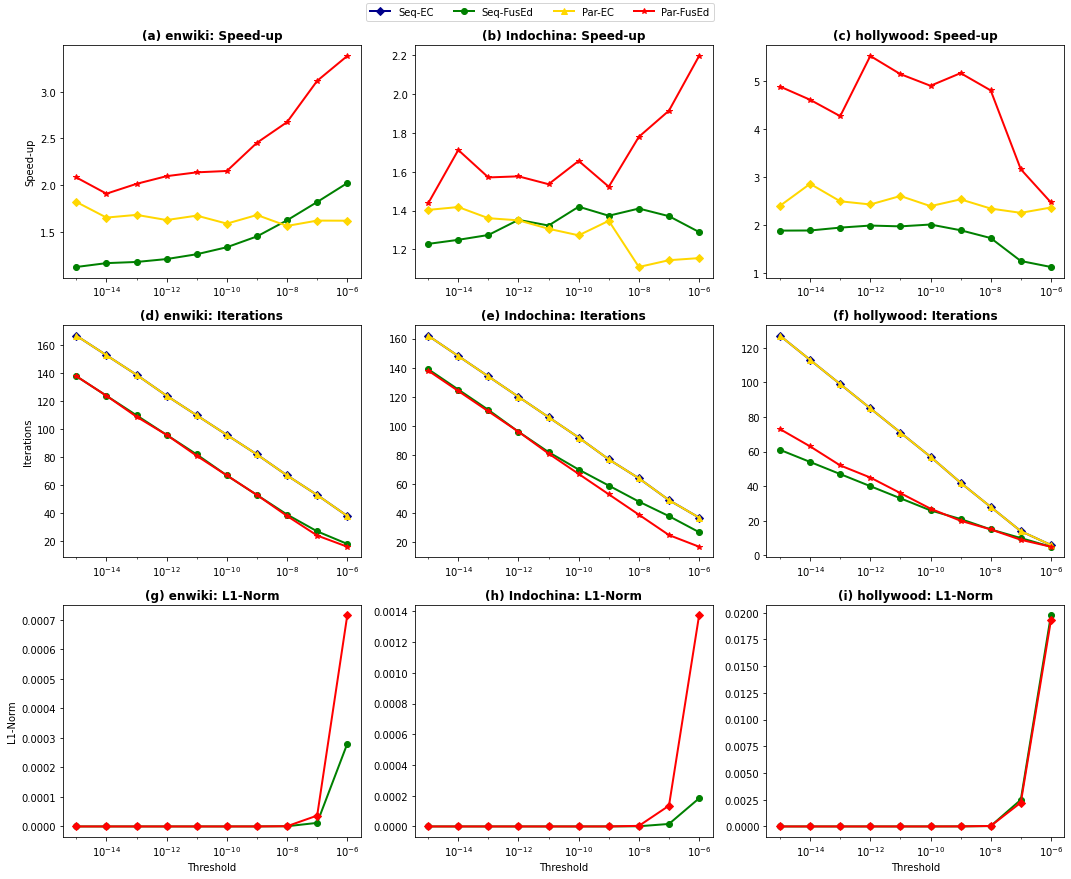}
    \caption{Detailed analysis of enwiki-2013, indochina-2004 and hollywood-2011 benchmark for speedup, number of iterations and L1-Norm with varying threshold ranging from $10^{-15}$ to $10^{-6}$ and number of threads fixed to 10 for parallel versions}
    \label{fig:combinedPlots}
\end{figure}

In the Seq-FusEd algorithm, nodes follow strict index ordering in computing the PageRank values. Suppose if $u$ and $v$ are two consecutive nodes with index values $i, j$ respectively, and $i < j$, then node $u$'s PageRank computation happens first, followed by node $v$. This property restricts the edge-centric PageRank algorithm to take the advantage of Gauss-Seidel approximation fully. 
In the Par-FusEd algorithm, we do not preserve the order for PageRank computations. This property will lead to asynchronicity within each iteration of PageRank computations. This behaviour is already studied in~\cite{Silvestre2018APA}, where they term it as an asynchronous Gauss-Seidel approximation. This technique results in faster convergence and improves speed-up over other versions.

Figure~\ref{fig:combinedPlots}d, \ref{fig:combinedPlots}e and \ref{fig:combinedPlots}f show the number of iterations taken by each program to terminate. With the lower threshold value, we obtain the higher precision PageRank values. With the decrease in the threshold value, the precision of each node PageRank value improves, which increases the number of iterations to terminate. We observe that Seq-FusEd and Par-FusEd algorithms take fewer iterations to converge on all the benchmarks as both algorithms exhibit approximate behaviour.

Figure~\ref{fig:combinedPlots}g, \ref{fig:combinedPlots}h and \ref{fig:combinedPlots}i show the variation in \textit{L1-Norm} with the change in threshold. 
Higher threshold value results in higher \textit{L1-Norm} for Seq-FusEd and Par-FusEd algorithms. With the increase in the threshold value, the \textit{L1-Norm} value increases. After the threshold goes below $10^{-8}$, the error value becomes negligible for all three benchmarks. We compute the \textit{L1-Norm} by calculating the Manhattan distance between Seq-EC and our proposed \fecpra{} algorithm. 
We conclude that fusing the two loops from Algorithm-1 will lead to better performance with lower threshold values.
Also, this technique is helpful when we have higher precision requirements.

\begin{figure}[h]
    \centering
    \includegraphics[scale=0.3]{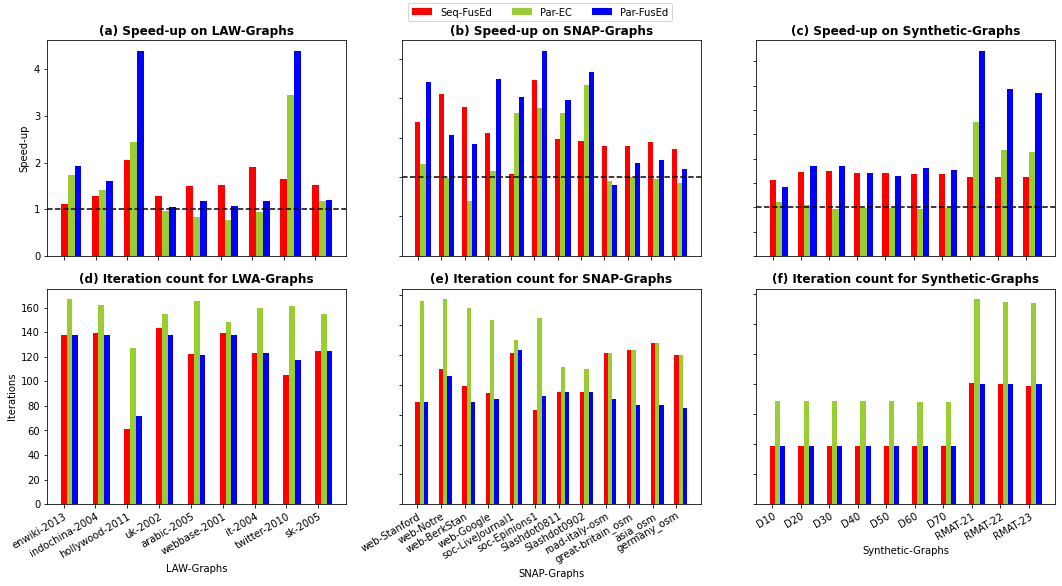}
    \caption{Detailed analysis of LAW, SNAP and Synthetic datasets for speedup and iterations with threshold value fixed to $10^{-15}$ and number of threads fixed to 10 for parallel versions.}
    \label{fig:4}
\end{figure}

For the rest of the experiments, we fix the threshold value to $10^{-15}$. Figure~\ref{fig:4}a,~\ref{fig:4}b and~\ref{fig:4}c show the speedup obtained by Par-EC, Seq-FusEd and Par-FusEd algorithms wrt Seq-EC algorithm on LAW~\cite{LAW}, SNAP~\cite{snapnets} and Synthetic datasets~\cite{DBLP:conf/sdm/ChakrabartiZF04} respectively. We observe that \fecpra{} algorithm performs better when compared with baseline algorithm for synthetic benchmarks. The results also show some performance decrease of our proposed fusion technique on specific benchmarks due to improper scheduling of the thread computations, which we will address in our future work. 

Figure~\ref{fig:4}d,~\ref{fig:4}e and~\ref{fig:4}f corresponds to the number of iterations taken by each version to converge. In LAW~\cite{LAW} and Synthetic benchmarks~\cite{DBLP:conf/sdm/ChakrabartiZF04}, Seq-FusEd and Par-FusEd take fewer iterations to converge than the Par-EC algorithm for all datasets. Since Seq-EC and Par-EC will always take the same number of iterations to converge, we did not include Seq-EC results to simplify the plots. In web-networks and social-networks from the SNAP~\cite{snapnets} benchmark, iterations variation is similar to LAW/Synthetic. Seq-FusEd takes the same number of iterations as the Par-EC algorithm for road-networks, and Par-FusEd takes fewer iterations to converge.

\begin{figure}[h]
    \centering
    \includegraphics[scale=0.38]{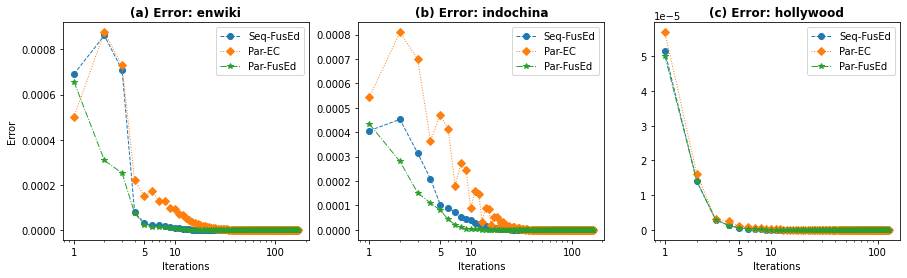}
    \caption{Detailed analysis of enwiki-2013, indochina-2004 and hollywood-2011 benchmark for Error variation with threshold value fixed to $10^{-15}$ and number of threads fixed to 10 for parallel versions}
    \label{fig:error_plots}
\end{figure}

Figure~\ref{fig:error_plots}a,~\ref{fig:error_plots}b and~\ref{fig:error_plots}c show the decrease in error value with the increase in iteration number on \textit{enwiki-2013}~\cite{LAW}, \textit{indochina-2004}~\cite{LAW} and \textit{hollywood-2011}~\cite{LAW} benchmarks, respectively. The Gauss-Seidel property for our proposed \fecpra{} increases convergence rate of the algorithm. Especially for Par-FusEd, the asynchronicity of Gauss-Seidel approximation in each iteration will allow more number of nodes to access the contribution values computed from current iteration, resulting in better convergence than Seq-FusEd. 
We conclude that our proposed \fecpra{} algorithm convergence rate is high and converges faster.

\begin{figure}[h]
    \centering
    \includegraphics[scale=0.38]{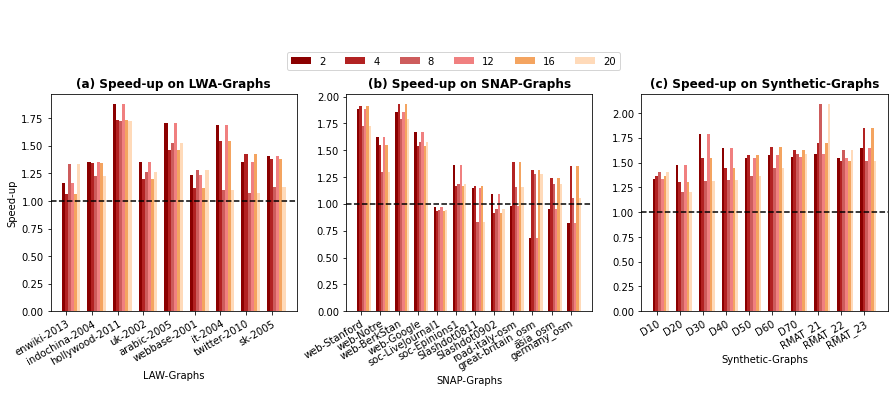}
    \caption{Detailed analysis of LAW, SNAP and Synthetic graphs for speedup on Par-FusEd wrt Par-EC with threshold value fixed to $10^{-15}$ and number of threads varying from \{2, 4, 8, 12, 16, 20\}}
    \label{fig:analysis_WrtBaseline}
\end{figure}

The improvement in speedup on Par-FusEd wrt Par-EC for various thread numbers is shown in Figure~\ref{fig:analysis_WrtBaseline}. For all the thread numbers varying from $\{2, 4, 8, 12,\\ 16, 20\}$, we see that \fecpra{} algorithm always performs better than the Edge-Centric baseline algorithm for LAW, SNAP and Synthetic graphs shown in Figure~\ref{fig:analysis_WrtBaseline}a, \ref{fig:analysis_WrtBaseline}b and \ref{fig:analysis_WrtBaseline}c respectively.
In LAW and Synthetic Graphs, varying the thread number will always leads to performance improvement.
For SNAP datasets, except for soc-Livejournal1, all other datasets are showing significant speed-up with varying threads.
The performance improvement with varying thread number is not uniform because of inefficient scheduling of data to each thread. We consider the scheduling problem as our future work.

\section{Conclusion and Future work}
\label{sec:conclusion}
This paper proposes the compiler-optimization based approximate technique, \fecpra{} algorithm that improves locality and reduces irregular memory accesses. Applying Loop-Fusion on Edge-centric PageRank algorithm with irregular memory accesses leads to \textit{faster convergence} in significantly \textit{less time} with negligible error in \textit{L1-norm} on both standard and synthetic benchmarks.

Our future work focuses on data scheduling for proper load-balancing and applying other compiler optimization techniques like loop-perforation to improve our proposed \fecpra{} algorithm on varied graph structures and storage formats. 
Although \fecpra{} is limited to shared memory architecture, it can be extended to the distributed environment with minimalist algorithmic changes.
Further, we plan to develop non-blocking PageRank variant while incorporating these techniques.

\bibliographystyle{splncs04}
\bibliography{References}

\end{document}